\documentclass[11pt]{amsart}
\usepackage{amssymb,latexsym}

\theoremstyle{definition}
\newtheorem{theorem}{Theorem}
\newtheorem{corollary}[theorem]{Corollary}
\newtheorem{proposition}[theorem]{Proposition}
\newtheorem{lemma}[theorem]{Lemma}
\newtheorem{definition}[theorem]{Definition}
\newtheorem{example}[theorem]{Example}
\newtheorem{notation}[theorem]{Notation}
\newtheorem{remark}[theorem]{Remark}
\newtheorem{question}[theorem]{Question}
\newtheorem{conjecture}[theorem]{Conjecture}

\usepackage{multirow}
\usepackage{graphicx}
\usepackage{tikz}
\usetikzlibrary{matrix,decorations.pathreplacing}

\newcommand{\numberset}{\mathbb}

\newcommand{\F}{\numberset{F}}

\newcommand{\Mat}{\mbox{Mat}}

\newcommand{\mC}{\mathcal{C}}
\newcommand{\mP}{\mathcal{P}}
\newcommand{\mG}{\mathcal{G}}

\usepackage{hyperref}
\newcommand{\mA}{\mathcal{A}}

\usepackage[margin=2.9cm]{geometry}

\newcommand{\mF}{\mathcal{F}}

\usepackage{amsaddr}

\begin{document}

\title{Subspace codes from Ferrers diagrams}

\author{Elisa Gorla$^1$ and  Alberto Ravagnani$^2$}
\address{Institut de Math\'{e}matiques
, Universit\'{e} de Neuch\^{a}tel\\Rue
Emile-Argand 11, CH-2000 Neuch\^{a}tel, Switzerland}
\email{$^1$elisa.gorla@unine.ch}
\email{$^2$alberto.ravagnani@unine.ch}


\thanks{The authors were partially supported by the Swiss National Science Foundation through grant no. 200021\_150207.}

\begin{abstract}
In this paper we give new constructions of Ferrer diagram rank metric codes, which achieve the largest possible dimension. 
In particular, we prove several cases of a conjecture by T. Etzion and N. Silberstein. 
We also establish a sharp lower bound on the dimension of linear rank metric anticodes with a given profile.
Combining our results with the multilevel construction, we produce examples of subspace codes with the largest known cardinality for the given parameters.
\end{abstract}

\maketitle

\section*{Introduction}\label{intr0}

Network coding is a branch of information theory concerned with data
transmission over noisy networks. A typical scenario where network coding 
techniques may be applied is multicast, i.e., one source
transmitting messages to multiple sinks through a network of intermediate nodes. 
Network coding provides techniques for optimizing the transmission and storage of data 
in digital file distribution, streaming television, peer-to-peer networking, and distributed storage, 
among other applications. 

In \cite{origine2} and \cite{origine3} it is
proved that maximal communication rate in multicasting transmission can be achieved by
allowing the intermediate nodes to perform random linear combinations of the
inputs that they receive, provided that the cardinality of the ground field is large enough.
Motivated by this result, in~\cite{KK1} R. K\"{o}tter and F. R. Kschischang
define a subspace code of constant dimension $k$ and length $n$ over the finite field $\F_q$ as a
collection of $k$-dimensional subspaces of $\F_q^n$. The
novel framework developed in \cite{KK1} also translates errors and
erasures correction in network communications into the problem of finding a
matrix of a given form, satisfying a rank constraint. 

Linear spaces of matrices with rank bounded below by a given $\delta$ and a fixed Ferrers diagram as
shape are introduced in \cite{ref1}, where they are used for constructing subspace codes. Since the 
cardinality of the subspace codes obtained increases with the dimension of the linear spaces of matrices, 
it is natural to ask what the maximum possible dimension of such linear spaces is. In
\cite{ref1} T. Etzion and N. Silberstein derive an upper bound in
terms of invariants of the Ferrers diagram, and conjecture that the bound is
sharp over finite fields. Some examples of Ferrers diagrams that provide evidence 
for the conjecture may be found in \cite{ref1} and \cite{alns}. 
In this paper, we establish the conjecture in several cases, including the most relevant case of ``large''
Ferrer diagrams (see Theorem~\ref{teo1} for a precise statement).

We also study the natural dual problem of linear spaces of matrices with rank bounded above by $\delta-1$
and a given profile $\mathcal{P}$ as shape (see Section~\ref{minore} for the definition of profile). 
Such spaces also appear in the literature under the name of linear anticodes. 
In this paper, we determine the largest possible dimension for any profile and over any field, 
and as a consequence we obtain an upper bound on the dimension of a linear space of matrices with rank bounded below by $\delta$
and a given profile $\mathcal{P}$ as shape.

Finally, we discuss how to apply our results to construct subspace codes over $\F_q$ of the largest 
known cardinality, for many choices of the parameters and for arbitrary $q$. 
Notice that the network coding scheme proposed in \cite{origine3} 
(which inspires the definition of subspace code) asymptotically achieves the maximum 
communication rate only for $q$ large enough, from which our interest in finding 
constructions for all values of $q$.
In addition, we show in an example that using constant weight lexicodes in the multilevel construction 
of~\cite{ref1} may not the best choice, in contrast to what is suggested in previous works.

The paper is organized as follows. In Section \ref{prel} and \ref{cose} we
give some preliminary definitions and results.  In Section \ref{varicasi} we prove
several relevant cases of the conjecture by T. Ezion and N. Silberstein.  Section \ref{minore} is concerned
with the study of the corresponding linear anticodes, for which we compute the maximum dimension for all
possible shapes and over any field. Finally, in Section \ref{esempibelli} we
show how the results from Section~\ref{varicasi} together with the multilevel 
construction of~\cite{ref1} provide new lower bounds on the size of subspace codes. 

New constructions of maximum dimensional Ferrer rank metric codes were obtained independently 
by A. Wachter-Zeh and T. Etzion in~\cite{tuvi}. In particular, their Construction~1 is the same as the 
construction which appears in the proof of our Theorem~\ref{teo3}.

\section{Preliminaries} \label{prel}

Throughout this paper, we denote by $[k]$ the set $\{1,\ldots,k\}$.
A set of $k \times m$ matrices over a field is said to be  a
$\underline{\delta}$-space
(resp., a $\overline{\delta}$-space) if it is a linear space and every non-zero
element of $V$ has rank at
least $\delta$ (resp., at most $\delta$). In this paper we study linear
spaces of matrices whose shape is a Ferrers diagram in the sense of \cite{ref1}.
For the convenience of the reader, we recall the definitions and results that we will use.

\begin{definition} \label{ferr}
 Given positive integers $k$ and $m$, a \textbf{Ferrers diagram} $\mF$ of \textbf{size}
$k \times m$ is a subset of $[k] \times [m]$ with the following properties:
\begin{enumerate}
  \item if $(i,j) \in \mF$ and $i > 1$, then $(i-1,j) \in \mF$,
 \item if $(i,j) \in \mF$ and $j < m$, then $(i,j+1) \in \mF$.
\end{enumerate}

For any $1 \le i \le k$, the $i$-th \textbf{row} of $\mF$ is the set of $(i,j)\in\mF$ with $j
\in [m]$. Similarly, for any $1 \le j \le m$ the
$j$-th \textbf{column} of $\mF$ is the set of $(i,j)\in\mF$ with $i\in [k]$.
\end{definition}

Notice that we do not require $(1,1)\in\mF$ or $(k,m)\in\mF$.

\begin{notation}
 We often identify a Ferrers diagram $\mF$ with the
cardinalities of its
rows. Indeed, given positive integers $m$, $k$ and $m
\ge r_1 \ge r_2
\ge
\cdots \ge r_k \ge 0$, there exists a unique Ferrers diagram $\mF$ of size $k
\times m$  such that the $i$-th row of $\mF$ has cardinality $r_i$ for any
$1 \le i \le k$. 
In this case we write $\mF=[r_1,...,r_k]$.
\end{notation}

\begin{remark} \label{caratt}
 Let $\mF=[r_1,...,r_k]$ and $\mF'=[r_1',...,r_k']$ be Ferrers diagrams of size
$k\times m$. We have $\mF' \subseteq \mF$ if and only if $r_i' \le r_i$ for all
$i=1,...,k$.
\end{remark}

Ferrers diagrams may be graphically represented as rows of
right-justified dots of decreasing cardinalities.  If
$\mF=[r_1,...,r_k]$, the first row of the graphical
representation of $\mF$ contains $r_1$ dots, the second row $r_2$ dots, and so on.

\begin{example} \label{ex1}
 Let $\mF:=[6,3,2,2]$ be a Ferrers diagram of size $6 \times 4$. The graphical
representation of $\mF$ is as follows:
$$\begin{array}{cccccc} \mbox{\textbullet} & \mbox{\textbullet} &
\mbox{\textbullet} & \mbox{\textbullet} & \mbox{\textbullet} &
\mbox{\textbullet} \\

 & & & \mbox{\textbullet} & \mbox{\textbullet} & \mbox{\textbullet} \\
 
 & & & &  \mbox{\textbullet} & \mbox{\textbullet} \\

 & & & & \mbox{\textbullet} &  \mbox{\textbullet}
\end{array}$$
\end{example}

\begin{definition}
 Let $M=(M_{i,j})$ be a $k \times m$ matrix. The \textbf{support} of $M$ is
  the set of its non-zero entries, i.e., $\mbox{supp}(M):=\{
(i,j) \in [k] \times [m] | M_{i,j} \neq 0 \}$.
 Let $\mF$ be a Ferrers diagram of size $k \times m$. We say that $M$ has
\textbf{shape} $\mF$ if $\mbox{supp}(M) \subseteq \mF$. 
\end{definition}

\begin{example}
 The two $4 \times 6$  matrices over $\F_2$
 $$\begin{bmatrix} 1 & 0 & 0 & 1 & 1 & 1 \\ 0 & 0 & 0 & 0 & 1 & 1 \\ 0 & 0 & 0 &
0  & 1 & 1  \\  0 & 0 & 0 & 0 & 1 & 1  \end{bmatrix}, \ \ \ \ \begin{bmatrix} 0
& 0 & 0 & 1 & 1 & 0 \\ 0 & 0 & 0 & 1 & 0 & 0 \\ 0 & 0 & 0 &
0  & 0 & 1  \\  0 & 0 & 0 & 0 & 1 & 0  \end{bmatrix}$$
have shape $\mF:=[6,3,2,2]$.
\end{example}

\begin{notation}
Fix a Ferrer diagram $\mF$. The set of matrices with entries in a field $\F$ which have shape $\mF$ form
a $|\mF|$-dimensional $\F$-vector space, which we denote by $\F[\mF]$. Equivalently, $$\F[\mF]=\{M\in\Mat_{k\times m}(\F)\mid \mbox{supp}(M)\subseteq\mF\}.$$
\end{notation}

A main open problem from \cite[Section VI]{ref1} is the following.

\begin{question} \label{que}
 Given integers $1 \le \delta \le k \le m$ and a Ferrers diagram $\mF$ of
size $k \times m$, what is the largest possible dimension of a
$\underline{\delta}$-space of $k \times m$ matrices with shape $\mF$ and
entries in a finite field $\F_q$?
\end{question}

\begin{remark} \label{transp}
Up to a transposition, the assumption $k \le m$ in Question \ref{que}
is not restrictive. In the sequel we always work with Ferrers diagrams of size
$k \times m$ with $k\le m$. This is also the relevant case for network coding
applications.
\end{remark}

Notice that Question \ref{que} makes sense over any field $\F$. We will
show in Remark \ref{ac} that the answer depends on the choice of $\F$. 
We denote by $$\mbox{MaxDim}_{\underline{\delta}}(\mF,\F)=\max\{\dim V\mid V\subseteq\F[\mF] \mbox{ is a } \underline{\delta}\mbox{-space}\}$$
the largest possible dimension of a $\underline{\delta}$-space of $k
\times m$ matrices with shape $\mF$ and entries in $\F$.

\begin{notation} \label{notazT}
Given integers $1 \le \delta \le k \le m$, $0 \le i \le \delta-1$, and a
Ferrers diagram $\mF$ of size $k\times m$, we denote by $T_\delta(\mF,i)$ the
cardinality of the set obtained from $\mF$ by removing the
topmost $i$ rows and the rightmost $\delta-i-1$ columns. Moreover, we set
$$T_\delta(\mF):=\min_{0 \le i \le \delta -1} T_\delta(\mF,i).$$
One always has $T_1(\mF)=|\mF|=T_1(\mF,0)$.
\end{notation}

\begin{example}
Let $\mF:=[6,3,2,2]$. We have
$T_4(\mF,0)=3$, $T_4(\mF,1)=1$, $T_4(\mF,2)=2$,
$T_4(\mF,3)=2$. Hence $T_4(\mF)=1$. Similarly one can check that $T_3(\mF)=4$
and $T_2(\mF)=7$.
\end{example}

The following lemma collects some properties which will be useful in the sequel. 
The proof is straightforward.

\begin{lemma} \label{tutto}
Let $\F$ be a field, $\mF$ and $\mF'$ be Ferrers diagrams. Assume that $\mF' \subseteq \mF$. We
have:
\begin{enumerate}
 \item $T_\delta(\mF) \ge T_\delta(\mF')$,
 \item $\mbox{MaxDim}_{\underline{\delta}}(\mF,\F) \ge
\mbox{MaxDim}_{\underline{\delta}}(\mF',\F)$.
\end{enumerate}
\end{lemma}

The authors of \cite{ref1} prove that $T_\delta(\mF)$ is an upper bound for
$\mbox{MaxDim}_{\underline{\delta}}(\mF,\F)$ for any $\delta$. Moreover, they
conjecture that the bound is attained when the field $\F=\F_q$ is finite, 
for any choice of $\delta$ and $\mF$. Notice that while in
\cite{ref1} the upper bound is stated only for finite fields, the proof 
works over an arbitrary field. Here we state the result in the general form.

\begin{theorem}[\cite{ref1}, Theorem 1]\label{tbound}
 We have
$$\mbox{MaxDim}_{\underline{\delta}}(\mF,\F) \le T_\delta(\mF)$$ for any field
$\F$, any Ferrers diagram $\mF$, and any $\delta\geq 1$.
\end{theorem}

\begin{conjecture}[\cite{ref1}, Conjecture 1]\label{tconj}
When $\F=\F_q$ is a finite field, equality holds in Theorem \ref{tbound} 
for any choice of the parameters $q$, $\mF$ and $\delta$.
\end{conjecture}

\section{From unrestricted matrices to matrices of prescribed shape}
\label{cose}

A well-studied case of Question \ref{que} is when $\mF=[k]\times
[m]$. This is the case of \textbf{unrestricted matrices}, solved by
Delsarte in 1978. The same result was established in the context of network
coding by K\"otter and Kschischang.

\begin{theorem}[\cite{del1}, Theorem 5.4 and Theorem 6.3; \cite{KK1}, Theorem 14]\label{rettang}
 Let $1 \le \delta \le k \le m$ be integers. 
We have 
 $$\mbox{MaxDim}_{\underline{\delta}}([k] \times [m],\F_q)= m(k-\delta+1)$$ for any finite
field $\F_q$. In particular, Conjecture~\ref{tconj} holds for $\mF=[k]\times[m]$.
\end{theorem}

For any choice of the parameters $\delta$, $k$ and $m$, Delsarte
provides in \cite{del1} a
construction of a $\underline{\delta}$-space of $k\times m$ matrices of maximum
dimension $m(k-\delta+1)$ over $\F_q$.
The properties of finite fields play a central role in his argument. It is interesting to notice that the answer to
Question~\ref{que} (hence the validity of Conjecture~\ref{tconj}) depends on the choice of the field $\F$, while 
the upper bound $T_\delta(\mF)$ does not. To illustrate this phenomenon, we show that 
Theorem~\ref{rettang} is false over an algebraically closed field.

\begin{example} \label{ac}
Let $\F$ be an algebraically closed field. We have $\mbox{MaxDim}_{\underline{k}}(\mF,\F)
\le 1$ for any $k \ge 1$ and any Ferrers diagram $\mF$ of size $k \times k$.
By contradiction, assume that
$A$ and $B$ are linearly independent $k\times k$ invertible matrices over $\F$
that span a $\underline{k}$-space.
Since $\F$ is algebraically closed, the polynomial $\det(\lambda B
+A)\in \F[\lambda]$ has a root,
say $\overline{\lambda}\in\F$. As a consequence, $\overline{\lambda}A+B$
is not invertible. Since $A$ and $B$ span a $\underline{k}$-space, we conclude that
$\overline{\lambda}A+B=0$, a contradiction.
\end{example}

Hence we have shown the following.

\begin{proposition} \label{nota}
Let $\F$ be an algebraically closed field and $\mF=[r_1,\ldots,r_k]$ be a Ferrers diagram of size $k\times k$, such that $r_i\geq k-i+1$ for $i=1,\ldots,k$. Then $\mbox{MaxDim}_{\underline{k}}(\mF,\F)=1$. In particular,
Conjecture~\ref{tconj} and Theorem~\ref{rettang} do not hold over an algebraically closed field.
\end{proposition}

An easy way to obtain a $\underline{\delta}$-space of $k \times m$
matrices with
a given shape $\mF$ is the following. Take a $\underline{\delta}$-space of
unrestricted $k \times m$ matrices, and select the ones with the appropriate shape $\mF$. 
The dimension of the $\underline{\delta}$-space obtained can be lower bounded as follows.

\begin{proposition} \label{banale0}
Let $1 \le \delta \le k \le m$ be integers, and let $\mF$ be a Ferrers
diagram of size $k \times m$. Then for any field $\F$
$$\mbox{MaxDim}_{\underline{\delta}}(\mF,\F) \ge
\mbox{MaxDim}_{\underline{\delta}}([k]\times [m],\F)-km+|\mF|.$$ In particular, if $\F=\F_q$ we have
$$\mbox{MaxDim}_{\underline{\delta}}(\mF,\F_q) \geq
|\mF|-m(\delta-1).$$
\end{proposition}

\begin{proof}
Let $\F[\mF]$ be the $\F$-vector space of $k\times m$ matrices with
entries in $\F$ and shape $\mF$. Clearly, $\dim\F[\mF]=|\mF|$.
Consider a $\underline{\delta}$-space $V$ of $k\times m$ matrices of dimension
$\mbox{MaxDim}_{\underline{\delta}}([k]\times [m],\F)$.  Then 
\begin{equation}\label{intersect}
\dim V \cap \F[\mF]\geq \mbox{MaxDim}_{\underline{\delta}}([k]\times [m],\F)+|\mF|-km.
\end{equation}
If $\F=\F_q$, the inequality follows from (\ref{intersect}) and Theorem~\ref{rettang}.
\end{proof}

We then obtain the following easy consequence of Proposition~\ref{banale0}.

\begin{corollary}[\cite{ref1}, Theorem 2]\label{banale}
Let $1 \le \delta \le k \le m$ be integers, and let $\mF$ be a Ferrers
diagram of size $k \times m$ with $r_{\delta-1}=m$. Then Conjecture ~\ref{tconj} holds.
\end{corollary}

\begin{remark} \label{easy} Corollary~\ref{banale}
implies that Conjecture~\ref{tconj} holds for any 
Ferrers diagram, if $\delta=2$.
\end{remark}

\begin{remark}
The dimension of $V\cap \F[\mF]$ depends on the choice of $V$, where $V$ is a
$\underline{\delta}$-space of unrestricted matrices of maximum dimension. Let e.g. $\mF:=[3,2,1]$.
The linear spaces 
$$V_1:= \langle \begin{bmatrix} 1 & 0& 0 \\ 1& 1 & 0 \\ 0 & 1
& 1 \end{bmatrix},  \begin{bmatrix} 0 & 1 & 0 \\ 1& 1 & 1 \\ 1 & 1
& 0 \end{bmatrix},   \begin{bmatrix} 0 & 0 & 1 \\ 0 & 1 & 0 \\ 1 & 0
& 0 \end{bmatrix} \rangle , \ \ \ \ \ \  
V_2:= \langle  \begin{bmatrix} 1 & 0& 0 \\ 0& 1 & 0 \\ 0 & 1
& 1 \end{bmatrix},  \begin{bmatrix} 0 & 1 & 0 \\ 0& 0 & 1 \\ 1 & 1
& 1 \end{bmatrix},   \begin{bmatrix} 0 & 0 & 1 \\ 1 & 1 & 0 \\ 1 & 0
& 1 \end{bmatrix} \rangle$$
are both $\underline{3}$-spaces of unrestricted matrices over $\F_2$ of maximal
dimension $3$.  However we have
$$V_1 \cap \F_2[\mF]= \langle \begin{bmatrix} 1 & 1& 0 \\
0& 1 & 1
\\ 0 & 0
& 1 \end{bmatrix} \rangle \ \ \ \ \mbox{and} \ \ \ \ V_2 \cap
\F_2[\mF]= \{ 0 \} .$$
\end{remark}

\begin{remark}
 It is well-known that the rank distribution of a $\underline{\delta}$-space of maximum dimension of
unrestricted matrices with entries in a finite field is completely determined by $\delta$, $k$, and $m$ 
(see \cite{del1} or \cite{gui}). This is in general not the case for
$\underline{\delta}$-spaces of matrices with prescribed shape and maximum
dimension. For example, let $\mF:=[3,2,1]$. The two 
$\underline{2}$-spaces of matrices over $\F_2$
 $$W_1:= \langle \begin{bmatrix} 1 & 1 & 1 \\ 0 & 1 & 0 \\ 0 & 0 & 0 
\end{bmatrix} ,  \begin{bmatrix} 1 & 0 & 1 \\ 0 & 0 & 1 \\ 0 & 0 & 0 
\end{bmatrix},   \begin{bmatrix} 1 & 1 & 1 \\ 0 & 0 & 0 \\ 0 & 0 & 1 
\end{bmatrix}\rangle, \ \ \ \ \  W_2:= \langle \begin{bmatrix} 1 & 1 & 1 \\ 0 &
1 & 0 \\ 0 & 0 & 0 
\end{bmatrix} ,  \begin{bmatrix} 0 & 1 & 1 \\ 0 & 0 & 1 \\ 0 & 0 & 0 
\end{bmatrix},   \begin{bmatrix} 1 & 0 & 0 \\ 0 & 0 & 0 \\ 0 & 0 & 1 
\end{bmatrix}\rangle$$
have shape $\mF$ and maximum dimension 3. However, they have different rank distributions. 
\end{remark}

\section{Evidence for the conjecture} \label{varicasi}

In this section, we give explicit constructions of $\underline{\delta}$-spaces of matrices with prescribed shapes. 
This allows us to compute the value of $\mbox{MaxDim}_{\underline{\delta}}(\mF,\F)$ for many
choices of $\mF$ and $\F$. As a consequence, we establish several new cases of Conjecture~\ref{tconj}.

\begin{theorem} \label{teo1}
 Let $2 \le \delta \le k \le m$ be integers, and let $\mF=[r_1,...,r_k]$ be a
Ferrers diagram of size $k \times m$. Assume $r_{\delta-1} \ge k$. We have
$$\mbox{MaxDim}_{\underline{\delta}}(\mF,\F_q)=T_\delta(\mF)=\sum_{i=\delta}^k r_i$$ for any
finite field $\F_q$. In particular, Conjecture~\ref{tconj} holds.
\end{theorem}

\begin{proof}
 Define the Ferrers diagram of size $k\times r_{\delta-1}$
$$\mF':=[\underbrace{r_{\delta-1},...,r_{\delta-1}}_{\delta-1},r_\delta,r_{
\delta+1 } , ... , r_k ]\subseteq\mF.$$
Since $r_{\delta-1}
\ge k$, by Corollary~\ref{banale} there exists a
$\underline{\delta}$-space of matrices with shape $\mF'$ and dimension
$|\mF'|-r_{\delta-1}(\delta-1)=\sum_{i=\delta}^{k} r_i$. Hence we have
$$T_\delta(\mF')\geq\mbox{MaxDim}_{\underline{\delta}}(\mF',\F_q) \ge \sum_{i=\delta}^{k} r_i\geq T_\delta(\mF'),$$ 
where the first inequality follows from Theorem~\ref{tbound} and the last from the definition of $T_\delta(\mF')$. Hence
$\mbox{MaxDim}_{\underline{\delta}}(\mF',\F_q) = \sum_{i=\delta}^{k} r_i =
T_\delta(\mF')$. 
Therefore
\begin{equation}\label{equal}
\sum_{i=\delta}^k r_i\geq T_\delta(\mF)\geq \mbox{MaxDim}_{\underline{\delta}}(\mF,\F_q)\geq
\mbox{MaxDim}_{\underline{\delta}}(\mF',\F_q)=\sum_{i=\delta}^k r_i \ ,
\end{equation}
where the first inequality follows from the definition of $T_\delta(\mF)$, the second from
Theorem~\ref{tbound}, and the third from Lemma~\ref{tutto}.
Therefore all the inequalities in (\ref{equal}) are equalities.
\end{proof}

\begin{remark}
To construct applicable subspace codes using the multilevel construction, we usually need Ferrers diagrams with $m \gg k$. 
In addition, the vector spaces of matrices that contribute the most to the cardinality of the resulting subspace code 
correspond to Ferrers diagrams of large cardinality. Hence the case treated in Theorem~\ref{teo1} is most relevant in the applications.
\end{remark}

For some Ferrers diagrams of size $k \times k$, the maximum dimension of a $\underline{\delta}$-space 
of matrices can be lower-bounded as follows. 

\begin{theorem} \label{teo2}
 Let $k \ge 1$ be an integer, and let $\mF$ be a Ferrers diagram of size $k
\times k$. Assume that $k/2 \le T_k(\mF) \le k-1$. We have
$$\mbox{MaxDim}_{\underline{k}}(\mF,\F_q) \ge \max \left\{ 2T_k(\mF)-k+1,
\bigg\lfloor
\frac{k}{2} \bigg\rfloor \right\}.$$
In particular, Conjecture \ref{tconj} holds in the following cases:\begin{itemize} 
\item $\delta=k=m$ even and $T_k(\mF)=k/2$, 
\item $\delta=k=m$ and $T_k(\mF)= k-1$.
\end{itemize}
\end{theorem}

\begin{proof}
By definition of $T_k(\mF)$, both the first column and
the $k$-th row of $\mF$ have cardinality at least $t:=T_k(\mF)$. As a
consequence, $\mF$
contains the Ferrers diagram
$\mF':=[\underbrace{k,...,k}_{t},\underbrace{t,...,t}_{k-t}]$. 
Since $\mF'\subseteq\mF$ and $T_k(\mF')=t$, by Lemma~\ref{tutto} it suffices to
prove the thesis for $\mF'$. The graphical
representation of $\mF'$ is:

\begin{center}
\begin{tikzpicture}[decoration=brace]
    \matrix (m) [matrix of math nodes] {
     \mbox{\textbullet} &  \mbox{\textbullet} &  \mbox{\textbullet} & 
\mbox{\textbullet} &  \mbox{\textbullet} &  \mbox{\textbullet}  & 
\mbox{\textbullet}  &  \mbox{\textbullet} \\
        \mbox{\textbullet} &  \mbox{\textbullet} &  \mbox{\textbullet} & 
\mbox{\textbullet} &  \mbox{\textbullet} &  \mbox{\textbullet}  & 
\mbox{\textbullet}  &  \mbox{\textbullet} \\
        \mbox{\textbullet} &  \mbox{\textbullet} &  \mbox{\textbullet} & 
\mbox{\textbullet} &  \mbox{\textbullet} &  \mbox{\textbullet}  & 
\mbox{\textbullet}  &  \mbox{\textbullet} \\
        \mbox{\textbullet} &  \mbox{\textbullet} &  \mbox{\textbullet} & 
\mbox{\textbullet} &  \mbox{\textbullet} &  \mbox{\textbullet}  & 
\mbox{\textbullet} &  \mbox{\textbullet} \\
        \mbox{\textbullet} &  \mbox{\textbullet} &  \mbox{\textbullet} & 
\mbox{\textbullet} &  \mbox{\textbullet} &  \mbox{\textbullet}  & 
\mbox{\textbullet} &  \mbox{\textbullet} \\
  &  &   & 
\mbox{\textbullet} &  \mbox{\textbullet} &  \mbox{\textbullet}  & 
\mbox{\textbullet} &  \mbox{\textbullet} \\
&  &   & 
\mbox{\textbullet} &  \mbox{\textbullet} &  \mbox{\textbullet}  & 
\mbox{\textbullet} &  \mbox{\textbullet} \\
&  &   & 
\mbox{\textbullet} &  \mbox{\textbullet} &  \mbox{\textbullet}  & 
\mbox{\textbullet} &  \mbox{\textbullet} \\
    };
    \draw[decorate,transform canvas={xshift=-0.3em},thick] (m-5-1.south west)
-- node[left=2pt] {$t$} (m-1-1.north west);
    \draw[decorate,transform canvas={yshift=-1.3em},thick] (m-8-8.north east)
-- node[above=-16pt] {$t$} (m-8-4.north west);
\end{tikzpicture}
\end{center}
\vspace{0.5cm}

Let $k_1=\lfloor k/2 \rfloor$ and  $k_2= \lceil k/2\rceil$. We have $t \ge
k_2$ by assumption. 
By Theorem \ref{rettang} there exists a $\underline{k_1}$-space $V_1$
 (resp., a $\underline{k_2}$-space $V_2$) of $k_1 \times k_1$ (resp., $k_2
\times k_2$) matrices with entries in $\F_q$ of dimension $k_1$ (resp., $k_2$).
Let $\{ M_1,...,M_{k_1} \}$ be a basis of $V_1$ and let  $\{ N_1,...,N_{k_2} \}$ be
a basis of $V_2$. The matrices 
$$ H_i:= \begin{bmatrix}
   M_i & 0 \\
   0 & N_i
  \end{bmatrix}, \ \  i=1,...,k_1 $$
span a $\underline{k}$-space of matrices with entries
in $\F_q$ and shape $\mF'$, of dimension $k_1=\lfloor k/2 \rfloor$. 
Therefore $\mbox{MaxDim}_{\underline{k}}(\mF',\F_q) \ge \lfloor
k/2 \rfloor$. 

Let us prove that $\mbox{MaxDim}_{\underline{k}}(\mF',\F_q) \ge 2t-k+1$ . If $k=t+1$, then by Corollary~\ref{banale}
$\mbox{MaxDim}_{\underline{k}}(\mF',\F_q)\ge |\mF'|-k(k-1)=k-1=2t-k+1$. If $k \ge t+2$ , 
let $\{ 1,\alpha,...,\alpha^{k-1}\}$ be an $\F_q$-basis of $\F_{q^k}=\F_q(\alpha)$. 
For $0\leq i \leq 2t-k$ define the $\F_q$-linear map $$\begin{array}{rcl} f_i: \F_{q^k} & \to & \F_{q^k} \\
x & \mapsto & \alpha^ix. 
\end{array}$$
Let $W:=\mbox{Span}_{\F_q} \{f_0,\ldots,f_{2t-k}\} \subseteq
\mbox{Hom}_{\F_q}(\F_{q^k}, \F_{q^k}).$ Since $t<k$, then 
$\dim W = 2t-k+1$, and any $f \in W\setminus\{0\}$ is invertible. 
Moreover, the matrices associated to the elements of $W$ with respect to the basis 
$\{ \alpha^{k-1},\ldots,\alpha,1\}$ and putting the images in the rows have shape $\mF'$. 
In fact, for $0\leq i\leq 2t-k$ we have $f_i(\alpha^{k-j})=\alpha^{k+i-j}$ with $0\leq k+i-j\leq t-1$ for $t+1\leq j\leq k$.
This proves that $\mbox{MaxDim}_{\underline{k}}(\mF',\F_q) \geq\dim W=2t-k+1$.
\end{proof}

\begin{example}
Let $q:=5$, $k:=4$ and $\mF:=[4,4,2,2]$. We apply the first part of the
proof of Theorem~\ref{teo2} to construct a 2-dimensional 
$\underline{4}$-space of shape $\mF$. Let $V=V_1=V_2$ be the vector space generated over $\F_5$ 
by $$\begin{bmatrix} 0 & 1 \\
 3 & 1 \end{bmatrix} \ \mbox{ and } \ \begin{bmatrix} 3 & 1 \\
 3 & 4 \end{bmatrix}.$$ $V$ is a $\underline{2}$-space, hence the vector space generated by the two matrices
$$\begin{bmatrix} 0 & 1 & 0 & 0 \\ 3 & 1 & 0 & 0 \\ 0 & 0 & 0 & 1 \\  0 & 0 & 3
& 1 \end{bmatrix}, \ \  \begin{bmatrix} 3 & 1 & 0 & 0 \\ 3 & 4 & 0 & 0 \\ 0 & 0
& 3 & 1 \\  0 & 0 & 3
& 4 \end{bmatrix} $$
is a 2-dimensional $\underline{4}$-space.
\end{example}

\begin{remark}
 The lower bound of Theorem \ref{teo2} is not sharp for all choices of the parameters. 
Let e.g. $k:=5$, $q:=3$ and $\mF:=[5,5,5,3,3]$. We have $T_5(\mF)=3$, hence $$\max
\left\{ 2T_5(\mF)-5+1,
\lfloor 5/2 \rfloor \right\}=2.$$ On the other hand, the three matrices over
$\F_3$
$$\begin{bmatrix} 1 & 2 & 1 & 0 & 0 \\ 0 & 0 & 0 & 1 & 0 \\ 1 & 0 & 1 & 0 & 0 
\\ 0 & 0 & 0 & 1 & 1 \\ 0 & 0 & 1 & 0 & 0 \end{bmatrix}, \ \ \ \ \begin{bmatrix}
1 & 0 & 1 & 0 & 1 \\ 2 & 0 & 0 & 0 & 0 \\ 0 & 1 & 0 & 0 & 0 \\  0 & 0 & 1 & 1 &
0 \\ 0 & 0 & 0 & 1 & 0   \end{bmatrix}, \ \ \ \ \begin{bmatrix} 0 & 0 & 0 & 1 &
0 \\ 0 & 1 & 0 & 0 & 0 \\  1 & 0 & 0 & 0 & 0 \\ 0 & 0 & 1 & 0 & 0 \\ 0 & 0 & 0 &
0 & 1 \end{bmatrix} $$
span a 3-dimensional $\underline{5}$-space. Hence
$\mbox{MaxDim}_{\underline{5}}(\mF,\F_3)=3$.
\end{remark}

The remainder of this section is concerned with 
Ferrers diagrams with an ``upper triangular'' profile. We will give a
lower-bound on $\mbox{MaxDim}_{\underline{\delta}}(\mF,\F)$
in terms of the lengths of the diagonals of $\mF$, provided that the field 
$\F$ is large enough.
As a corollary, we compute the maximum possible dimension of
$\underline{\delta}$-spaces of upper triangular matrices over
sufficiently large fields, establishing Conjecture \ref{tconj} for some families of
diagrams. Before proving the next theorem, we recall some
elementary results from classical coding theory.

\begin{definition}
A \textbf{linear code} of lenght $n \ge 1$ and dimension $k$ over a field $\F$ is a
$k$-dimensional vector subspace of $\F^n$. The \textbf{weight} of a vector in
$\F^n$ is the number of its non-zero components. The \textbf{minimum distance}
of a non-zero linear code $C \subseteq \F^n$ is the minimum of the weights of 
the elements of $C \setminus \{ 0 \}$.
\end{definition}

\begin{lemma} \label{reed}
 Let $\F$ be a field. For any integers $1 \le \delta \le n$ there exists a
code $C \subseteq \F^n$ of minimum distance $\delta$ and dimension
$n-\delta+1$, provided that $|\F| \ge n-1$. 
\end{lemma}

\begin{proof}
If $|\F|=n-1$ the result follows from \cite[Chapter 11, Theorem 9]{MS}.
Now assume $|\F| \ge n$. 
Let $\alpha_1,...,\alpha_n \in \F$ distinct.
Denote by $\F[x]_{\le n-\delta}$ the $\F$-space of polynomials with
coefficients in $\F$ and degree at most $n-\delta$. The $\F$-linear map 
$\varphi: \F[x]_{\le n-\delta} \to \F^n$ defined by 
$\varphi(f)=(f(\alpha_1),...,f(\alpha_n))$ is injective by the Fundamental Theorem of Algebra. 
The image of $\varphi$ is a code with the expected properties.
\end{proof}

\begin{definition}
 Let $\mF$ be a Ferrers diagram of size $k\times m$. The $r$-th
\textbf{diagonal} of $\mF$ is the set of elements of $\mF$ of the form
$(i,j)$ with $i-j+m=r$. Notice that we enumerate diagonals from right to left.
Similarly, if $M$ is a matrix with shape $\mF$, we define the $r$-th
\textbf{$\mF$-diagonal} of $M$ as the vector with entries
$M_{i,i+m-r}$ such that $(i,i+m-r) \in \mF$. 
\end{definition}

\begin{example}
Let $\mF:=[4,2,2,1]$. The second diagonal of $\mF$ has cardinality two, the
third and
the fourth have cardinality three. Consider the matrix $M$ of shape
$\mF$ given by
$$M:= \begin{bmatrix}
   a & b & c & d \\ 0 & 0 & e & f \\ 0 & 0 & g & h \\ 0 & 0 & 0 & i
  \end{bmatrix}.
$$
The second $\mF$-diagonal of $M$ is $(c,f)$, the third is $(b,e,h)$, and
the
fourth is
$(a,g,i)$.
\end{example}

A similar construction to the one that we use to prove the next theorem appears in~\cite{roth}. 
We thank T. Etzion for bringing this work to our attention.

\begin{theorem} \label{teo3}
 Let $1 \le \delta \le k \le m$ be integers, and let $\mF$ be a Ferrers diagram
of size $k \times m$. Assume that $\mF$ has $n$ diagonals
$D_1,...,D_n$ of cardinality at least $\delta-1$. $D_i$ is the $\alpha_i$-th diagonal of $\mF$, for some $\alpha_1<\ldots<\alpha_n$.
If $|\F|\ge \max_{i=1}^n  |D_i|-1$, then
$$\mbox{MaxDim}_{\underline{\delta}}(\mF,\F) \ge \sum_{i=1}^n
(|D_i|-\delta+1).$$
\end{theorem}

\begin{proof}
First we notice that the summands corresponding to diagonals of cardinality 
$\delta-1$ give no contribution to the lower bound. Hence we may assume without loss of generality that 
$|D_i| \ge \delta$ for $i=1,...,n$. By Lemma \ref{reed}, for any $i=1,...,n$ there exists a code
$C_i \subseteq \F^{|D_i|}$ of minimum distance $\delta$ and dimension
$|D_i|-\delta+1$.
Given vectors $v_1,...,v_n$ of lengths $|D_1|,...,|D_n|$ respectively, denote
by $M(v_1,...,v_n,\mF)$ the unique $k\times m$ matrix with the following
properties:
\begin{enumerate}
 \item the shape of $M(v_1,...,v_n,\mF)$ is $\mF$,
 \item the vector $v_i$ is the $\alpha_i$-th $\mF$-diagonal of
$M(v_1,...,v_n,\mF)$,
 \item all the remaining entries of $M(v_1,...,v_n,\mF)$ are zero.
\end{enumerate}
We claim that the linear space $$V:=\mbox{Span}_\F \left\{ M(v_1,...,v_n,\mF) \
: \ (v_1,...,v_n)
\in C_1 \times \cdots \times C_n \right\}$$
is a $\underline\delta$-space of $k\times m$ matrices with shape
$\mF$, of dimension $\sum_{i=1}^n (|D_i|-\delta+1)$.  To compute $\dim V$, 
observe that the map $C_1
\times \cdots \times C_n \to V$ given by $(v_1,...,v_n) \mapsto
M(v_1,...,v_n,\mF)$ is an $\F$-isomorphism. Since
$\dim(C_i)=|D_i|-\delta+1$ for all $i$, then
$\dim V=\sum_{i=1}^n (|D_i|-\delta+1)$. It remains to show that an
arbitrary non-zero matrix in $V$ has rank at least
$\delta$. Fix $M \in V \setminus \{ 0 \}$, and let
$r$ denote the maximum integer such that the $r$-th diagonal of $M$ is non-zero.
By definition of $V$, we have $r=\alpha_i$ for some $i$. Since $C_i$ has minimum
distance $\delta$, the $r$-th diagonal of $M$ has at least
$\delta$ non-zero entries. By the maximality of $r$, the entries of $M$
which lie below such diagonal are all zero.  It is easy to see that a
matrix $M$ of this form has rank at least $\delta$.
\end{proof}

\begin{corollary}\label{teo4}
Let $1 \le \delta \le k \le m$ be integers, and let $\mF=[r_1,...,r_k]$ be a Ferrers diagram
of size $k \times m$. Assume $r_i \ge m-i+1$ for $i=1,...,\delta-1$ and 
$r_i \le m-i+1$ for $i=\delta,...,k$. We have
$$\mbox{MaxDim}_{\underline{\delta}}(\mF,\F)=T_\delta(\mF)$$
for any field $\F$ such that $|\F| \ge \max_{i=\delta}^m |D_i|-1$, where $D_i$
denotes the
$i$-th diagonal of $\mF$.
In particular, Conjecture~\ref{tconj} holds. 
\end{corollary}

\begin{proof}
Since $r_i \ge m-i+1$ for
$i=1,...,\delta-1$, we have $|D_\delta|,...,|D_m| \ge \delta-1$. By Theorem
\ref{teo3}, $\mbox{MaxDim}_{\underline{\delta}}(\mF,\F) \ge \sum_{i=\delta}^m
(|D_i|-\delta+1)$. By
Theorem \ref{tbound} and the definition of $T_\delta(\mF)$, it suffices
to prove that $T_\delta(\mF,\delta-1)=\sum_{i=\delta}^m (|D_i|-\delta+1)$. Since 
$r_i \le m-i+1$ for $i=\delta,...,k$, and $r_i \ge m-i+1$ for
$i=1,...,\delta-1$, when we remove from $\mF$ the first $\delta-1$ rows we
obtain a set of cardinality $\sum_{i=\delta}^m (|D_i|-\delta+1)$, as claimed.
\end{proof}

\begin{corollary} \label{teo5}
Let $1 \le \delta \le k$ be integers. The maximum dimension of a $\delta$-space
of $k \times k$ upper (or lower) triangular matrices over any field $\F$ is
$\binom{k-\delta+2}{2}$, provided that $|\F| \ge k-1$. In particular,
Conjecture~\ref{tconj} holds. 
\end{corollary}

\begin{proof}
The Ferrers diagram corresponding to upper triangular $k\times k$ matrices is
$\mF:=[k,k-1,...,1]$, which satisfies the assumptions of
Corollary \ref{teo4}. Hence we only need to check that
$T_\delta(\mF)=\binom{k-\delta+2}{2}$. Fix any $0 \le i \le \delta-1$. We have
$$T_\delta(\mF,i)=\sum_{j=k-i+1}^k j \ + \ \sum_{j=1}^{k-i} j \ \
=\sum_{i=1}^{k-\delta+1} j \ = \ \binom{k-\delta+2}{2}.$$
It follows that $T_\delta(\mF)=\binom{k-\delta+2}{2}$.
\end{proof}

\begin{remark}
 The requirement $|\F| \ge k-1$ in the statement of Corollary \ref{teo5} is not
necessary, in general, for the existence of a $\underline\delta$-space
of $k \times k$ upper triangular matrices of dimension
$\binom{k-\delta+2}{2}$. For example, the three upper triangular $4 \times
4$ matrices over $\F_2$
$$ \begin{bmatrix} 0 &  0 & 0 & 1 \\ 0 & 1 & 0 & 0 \\ 0 & 0 & 1 & 0 \\ 0 & 0 &
0 & 0 \end{bmatrix}, \ \ \ 
 \begin{bmatrix} 0 &  1 & 0 & 0 \\ 0 & 0 & 1 & 1 \\ 0 & 0 & 0 & 1 \\ 0 & 0 & 0
& 0\end{bmatrix}, \ \ \ 
 \begin{bmatrix} 1 &  0 & 1 & 0 \\ 0 & 1 & 0 & 1 \\ 0 & 0 & 0 & 0 \\ 0 & 0 & 0
& 1\end{bmatrix}$$
span a $3$-dimensional
$\underline{3}$-space.
\end{remark}

% $$
% \begin{array}{ccccccc}
%     \multicolumn{2}{c}{\multirow{2}{*}{\raisebox{-0mm}{\scalebox{2}{$M_i$}}}}
% &   0 & \cdots  & \cdots & \cdots & 0 \\ 
% & &  0  & \cdots & \cdots & \cdots & 0 \\
% 0 & 0 &  0  & \cdots & \cdots & \cdots & 0 \\
% \vdots & \vdots &  \vdots  & \vdots & \vdots & \vdots & \vdots \\
% 0 & \cdots &  \cdots  & \cdots & 0 & 0 & 0 \\
% 0 & \cdots & \cdots & \cdots & 0
% & \multicolumn{2}{c}{\multirow{2}{*}{\raisebox{-0mm}{\scalebox{2}{$N_i$}}}} \\
% 0 & \cdots & \cdots & \cdots & 0 & & 
%   \end{array}
% $$

\section{$\overline{\delta}$-spaces of maximum dimension} \label{minore}

The problem stated in Question \ref{que} has the following
natural dual version, in terms of $\overline{\delta}$-spaces of matrices. 
Notice that $\overline{\delta}$-spaces of matrices are by definition linear rank metric anticodes.

\begin{question} \label{que2}
Given integers $1 \le \delta \le k \le m$ and a Ferrers diagram $\mF$ of
size $k \times m$, what is the largest possible dimension of a
$\overline{\delta}$-space of $k \times m$ matrices with shape $\mF$ and
entries in a finite field $\F_q$?
\end{question}

\begin{definition}
Given positive integers $k$ and $m$, define a \textbf{profile} of size $k \times m$
as a subset $\mathcal{P} \subseteq [k]\times [m]$. 
For any $1 \le i \le k$, the $i$-th \textbf{row} of $\mathcal{P}$ is the set of $(i,j)\in\mathcal{P}$ with $j
\in [m]$. Similarly, for any $1 \le j \le m$ the
$j$-th \textbf{column} of $\mathcal{P}$ is the set of $(i,j)\in\mathcal{P}$ with $i\in [k]$.

A $k \times m$ matrix $M$ has \textbf{shape} $\mP$ when $\mbox{supp}(M) \subseteq \mP$. 
\end{definition}

Notice that Ferrers diagrams are examples of profiles. In this section, using an idea from \cite{Mesh},
 we answer the following generalization of Question \ref{que2}.

\begin{question} \label{que3}
 Given integers $1 \le \delta \le k \le m$ and a profile $\mP$ of
size $k \times m$, what is the largest possible dimension of a
$\overline{\delta}$-space of $k \times m$ matrices with shape $\mP$ and
entries in an arbitrary field $\F$?
\end{question}

Let $\F$ be a field and $\mP$ be a profile of size $k \times m$. We denote by
$$\mbox{MaxDim}_{\overline{\delta}}(\mP,\F)$$
the maximum dimension of a $\overline{\delta}$-space of $k \times m$ matrices with
entries in $\F$ and shape $\mP$.

\begin{notation} \label{notazT2}
Let $1 \le \delta \le k \le m$ be integers, and let $\mP$ be a profile of
size $k \times m$. Given subsets $I \subseteq [k]$,
$J \subseteq [m]$ such that $|I|+|J|=\delta-1$, we denote by $T_\delta(\mP,I,J)$ 
the cardinality of the set obtained from $\mP$ by removing the rows of index $i\in I$
and the columns of index $j\in J$. Moreover, we set 
$$T_\delta(\mP):=\min \left\{ T_\delta(\mP,I,J) \ | \ I \subseteq [k], J \subseteq [m] \mbox{ and }
|I|+|J|=\delta-1 \right\}.$$ Finally, recall that a \textbf{line} of a matrix is either a row,
or a column of the matrix.
\end{notation}

\begin{remark} \label{casospeciale}
 When $\mP=\mF$ is a Ferrers diagram, the definition of $T_\delta(\mF)$ given in Notation~\ref{notazT} and the definition of $T_\delta(\mP)$ 
given in Notation~\ref{notazT2} coincide.
\end{remark}

\begin{lemma} \label{maggiore}
 Let $1 \le \delta \le k \le m$ be integers, and let $\mP$ be a profile
of size $k \times m$. We have
$$\mbox{MaxDim}_{\overline{\delta-1}}(\mP,\F) \ge |\mP|-T_{\delta}(\mP)$$
for any field $\F$.
\end{lemma}

\begin{proof}
Choose $I \subseteq [k]$ and $J\subseteq [m]$ such that $|I|+|J|=\delta-1$ and
 $T_\delta(\mP,I,J)=T_\delta(\mP)$.
Let $$\mP'=\{(i,j)\in\mP\mid i\in I \mbox{ or } j\in J\}.$$
Because of the choice of $I$ and $J$,
$|\mP'|=|\mP|-T_\delta(\mP)$. Denote by $\F[\mP']$ the vector space of
$k \times m$ matrices over $\F$ with shape $\mP'$. We have $\dim_\F \F[\mP']=
|\mP'|=|\mP|-T_\delta(\mP)$. Since the support of any $M \in
\F[\mP']\subseteq\F[\mP]$ is contained in at most $\delta-1$ lines, we have $\mbox{rank}(M)\le \delta-1$.
Hence $\mbox{MaxDim}_{\overline{\delta-1}}(\mP,\F) \ge |\mP|-T_{\delta}(\mP)$, 
as claimed.
\end{proof}

It is now easy to prove  the following generalization of Theorem \ref{tbound}.

\begin{theorem}
Let $1 \le \delta \le k \le m$ be integers, and let $\mP$ be a profile of size $k \times m$.
For any field $\F$ we have $$\mbox{MaxDim}_{\underline{\delta}}(\mP,\F) \le T_\delta(\mP).$$
\end{theorem}

\begin{proof}
 Let $V$ be a
$\underline{\delta}$-space of matrices with shape $\mP$ of dimension
$\mbox{MaxDim}_{\underline{\delta}}(\mP,\F)$. Similarly, let $W$ be a 
$\overline{\delta-1}$-space
of matrices with shape $\mP$ of dimension
$\mbox{MaxDim}_{\overline{\delta-1}}(\mP,\F)$. Denote by $\F[\mP]$ the
$|\mP|$-dimensional $\F$-vector space of $k \times m$ matrices with
shape $\mP$ and entries in $\F$. We
have $V \cap W = \left\{ 0 \right\}$ and $V\oplus W \subseteq \F[\mP]$. 
By Lemma \ref{maggiore}, $\dim V \leq |\mP|-(|\mP|-T_{\delta}(\mP))$.
\end{proof}

\begin{remark}
By Lemma \ref{maggiore}, Conjecture \ref{tconj} can be restated as follows:
Over a finite field $\F_q$ and for any $\delta$, the vector space $\F_q[\mF]$ of matrices of fixed 
shape $\mF$ decomposes as
$$\F_q[\mF]=\underline{V} \oplus
\overline{V},$$ where $\underline{V}$ is a $\underline{\delta}$-space
and $\overline{V}$ is a $\overline{\delta-1}$-space. 
We stress that this is in general false when the underlying field is not finite 
(see Proposition~\ref{nota}).
\end{remark}

\begin{notation}
For integers $1 \le k \le m$, let $\prec$ denote the lexicographic order
on $[k] \times [m]$, i.e., $(i,j) \prec (i',j')$ if and only if either $i<i'$ or
$i=i'$ and $j<j'$.
For a $k\times m$ matrix $M$ over a field $\F$
we set $$p(M):=\min \{ (i,j) \ | \  M_{i,j} \ne 0\}.$$ For a set
$\mathcal{A}$ of $k\times m$ matrices  define the $0$-$1$
matrix $M(\mathcal{A})$ over $\F$ as follows:
\begin{enumerate}
 \item $M(\mathcal{A})_{i,j}=1$ if $(i,j)=p(A)$ for some $A\in\mathcal A$,
 \item $M(\mathcal{A})_{i,j}=0$ otherwise.
\end{enumerate}
Finally, denote by $\rho(\mathcal{A})$ the minimal cardinality of a set of lines of
$M(\mathcal{A})$ which contain all the $1$'s appearing in $M(\mathcal{A})$.
\end{notation}

\begin{lemma} \label{le}(\cite{Mesh}, Theorem~1)
 Let $\mathcal{A}$ be a set of $k\times m$ matrices over a
field $\F$. Then $\mbox{Span}_\F(\mathcal{A})$ contains a matrix of rank at
least $\rho(\mA)$.
\end{lemma}

%\begin{proof}
% Without loss of generality we assume $k \le m$. If $M$ is any $k \times m$
%matrix, denote by $\overline{M}$ the $m
%\times m$ matrix obtained from $M$ by adding $m-k$ zero rows in the
%bottom. Set $\overline{\mathcal{A}}:= \{ \overline{M_1},...,\overline{M_r} \}$. 
% We have $M(\overline{\mathcal{A}})=\overline{M(\mathcal{A})}$, and so
%$\rho(\mA)=\rho(\overline{\mA})$. Apply \cite{Mesh}, Theorem 1.
%\end{proof}

The following theorem provides an answer to Question~\ref{que2} 
and Question~\ref{que3}. It is inspired by Theorem~2
of~\cite{Mesh}.

\begin{theorem} \label{teo6}
 Let $1 \le \delta \le k \le m$ be integers, and let $\mP$ be a profile
of size $k \times m$. We have
$$\mbox{MaxDim}_{\overline{\delta-1}}(\mP,\F) = |\mP|-T_{\delta}(\mP)$$
for any field $\F$. 
\end{theorem}

\begin{proof}
 By Lemma \ref{maggiore} it suffices to show that
$\mbox{MaxDim}_{\overline{\delta-1}}(\mP,\F) \le |\mP|-T_{\delta}(\mP)$.
Let $V$ be a $\overline{\delta-1}$-space of $k \times m$ matrices over $\F$ with
shape $\mP$ of dimension $r:=\mbox{MaxDim}_{\overline{\delta-1}}(\mP,\F)$.
Choose a basis $\{N_1,...,N_r\}$ of $V$. Let $\varphi$ be the
$\F$-isomorphism that sends a $k\times m$ matrix $M$ to the vector of length $km$ whose entries are the entries of $M$ ordered lexicographically. Define $w_i:=\varphi(N_i)$ for $i=1,...,r$. Perform Gaussian
elimination on $w_1,...,w_r$ and get vectors $v_1,...,v_r$. Set
$M_i:=\varphi^{-1}(v_i)$ for $i=1,...,r$. It is clear that $\mathcal{A}:=\{
M_1,...,M_r\}$ is a basis of $V$. Since $p(M_i)\neq p(M_j)$ for $i \neq j$, the support $\mP'$ 
of $M(\mA)$ has cardinality exactly $r$. 

Since $V$ is a $\overline{\delta-1}$-space, by Lemma \ref{le} the support
$\mP'$ is contained in a set of $i$ rows and $\delta-i-1$
columns for some $0 \le i \le \delta-1$. Since $\mP' \subseteq \mP$, 
we conclude that $|\mP'|
\le |\mP|-T_\delta(\mP)$.
\end{proof}

\section{Applicatons and examples} \label{esempibelli}

In this section we show how one can apply the results of Section \ref{varicasi}
to construct large subspace codes with given parameters,
for any size of the ground field $\F_q$. In particular we show how to construct 
the largest known codes for $q \ge 3$ and many choices of the parameters.
Moreover, being systematic, the constructions that we propose may be useful for 
designing efficient decoding algorithms. 

For $q=2$, $\delta=2,3$ and small values of $n$ and $k$, there exists subspace
codes which have larger cardinality than the codes we can construct using the results
contained in this paper (see e.g. \cite{e-v}, \cite{alns} and \cite{koku}). 
The techniques employed to produce such codes include a computer search,
which is not feasible for large values of $q$ and of the other parameters.

\begin{definition}
Given $k$-dimensional vector subspaces 
$X,Y \subseteq \F_q^n$,  the \textbf{injection distance} between $X$ and $Y$ is defined as $d_I(X,Y):=k-\dim(X \cap Y)$. 
Denote by $\mG_q(k,n)$ the set of $k$-dimensional subspaces of $\F_q^n$. The \textbf{minimum distance}
of a subspace code $\mathcal{C} \subseteq \mG_q(k,n)$ with
$|\mC| \ge 2$ is defined as the minimum of all the pairwise distances between
distinct elements of $\mC$. 
\end{definition} 

Let us briefly recall the \textbf{multilevel construction} for subspace codes proposed by T. Etzion
and N. Silberstein in \cite{ref1}. 

\begin{notation}
Let $X$ be a $k$-dimensional subspace of $\F_q^n$ and let $M(X)$
be the unique $k \times n$ matrix in row-reduced echelon form with rowspace $X$.
We associate to $X$ the binary vector $v(X)$ of length $n$ and weight $k$, which has a $1$ 
in position $i$ if and only if $M(X)$ has a pivot in the $i$-th column.
The vector $v(X)$ is called the \textbf{pivot vector} associated to $X$ and $M(X)$.
\end{notation}

\begin{lemma}[\cite{ref1}, Lemma 2]
Let $X,Y \in \mG_q(k,n)$. We have $d_I(X,Y) \ge \frac{1}{2} d_H(v(X),v(Y))$, where $d_H$ denotes the Hamming distance.
\end{lemma}

\begin{notation} \label{assoc}
Let $v$ be a binary vector of length $n$ and weigth $k$, and let $1 \le p_1 <
p_2 < \cdots < p_k \le n$ be the positions of the
$k$ ones of $v$. The \textbf{Ferrers diagram associated to $v$} is the Ferrars diagram $\mF_v=[r_1,\ldots,r_k]$
of size $k \times (n-k)$ with $r_i=n-k-p_i+i$ for all $i=1,...,k$.
\end{notation}

The following result is straightforward. See \cite{ref1}, Section III and IV
for examples and details.

\begin{lemma} \label{rref}
Let $v$ be a binary vector of length $n$ and weigth $k$, and let $1 \le p_1 <
p_2 < \cdots < p_k \le n$ be the positions of the
$k$ ones of $v$. Let $M \in \F_q[\mF_v]$. For
$j=1,...,n$ define $n_j:=\#\{1\leq i\leq k\mid p_i\leq j\}$.
There exists a unique  $k \times n$ matrix $N$ over $\F_q$ in
row-reduced echelon form having $v$ as pivot vector and
 $N_{i,j}=M_{i,j-n_j}$ for
all $i \in \{1,...,k\}$ and $j \in \{1,...,n\} \setminus  \{p_1,...,p_k\}$. 
\end{lemma}

We denote the matrix $N$ of Lemma \ref{rref} by $N(v,M)$. The
multilevel construction of \cite{ref1} is summarized in the following
result.

\begin{theorem}[\cite{ref1}, Theorem 3] \label{boundmulti}
Let $C$ be a binary code of constant weight $k$, length $n$ and minimum distance
at least $2\delta$. For any $v \in C$ let $S(v) \subseteq \F_q[\mF_v]$ be
a $\underline{\delta}$-space. The set
$$\left\{   \mbox{rowsp} \ N(v,M) \ | \ v \in C, M \in S(v) \right\} \subseteq
\mG_q(k,n)$$
is a subspace code of minimum distance al least $\delta$ and cardinality
$\sum_{v \in C} q^{\dim S(v)}$.
\end{theorem}

\begin{remark}
Large subspace codes with $\delta>2$ were obtained in \cite{ref1} combining the
multilevel construction and a computer search, for small values of $q$.
The computer search part is employed to find large spaces of matrices
of rank $\geq\delta$ and given shape. The results of Section~\ref{varicasi} 
allow us to construct in a systematic way (i.e., without a computer search) 
linear spaces of matrices with the same parameters as those found via 
computer search in~\cite{ref1}. In particular, we can construct subspace 
codes with the same parameters for any $q$. 
%why $q\geq 3$?
\end{remark}

\begin{remark}
 In \cite{aj}, A-L. Trautmann and J. Rosenthal propose the
\textbf{pending dots} construction to improve the multilevel
construction of \cite{ref1}. As the multilevel construction, the pending dots construction also depends
on the existence of large spaces of matrices with bounded rank and given shape. 
Using the idea of pending dots,  A-L. Trautmann and N.
Silberstein construct large subspace codes in
$\mG_q(k,n)$ of
minimum injection distance 
%$\delta=2$ (Section IV of \cite{alns} and Section IV of \cite{alns2}) and 
$\delta=k-1$ (see Section V of \cite{alns} and Section III of \cite{alns2}) for
arbitrary values of $q$. The Ferrers 
%diagrams they consider for the case $\delta=2$ are covered by Remark \ref{easy}, while the 
diagrams that they consider for the case $\delta=k-1$  (\cite{alns}, Lemma 23 and
\cite{alns2}, Lemma 18) are
special cases of the diagrams studied in Theorem \ref{teo3}.
\end{remark}

\begin{remark}
 Spread and partial spread codes (\cite{GR}, \cite{GR2},
\cite{partial}) can be obtained through the
multilevel construction for a special choice of the pivot vectors. 
The Ferrers diagrams associated to those pivot vectors are 
studied in Theorem \ref{rettang}.
\end{remark}

We now give some examples of how to combine the results of Section \ref{varicasi} 
and the multilevel construction to obtain subspace codes with the largest known 
cardinality for given $k$, $n$, and $\delta$.

\begin{example} \label{exB}
Let $(n,k,\delta):=(10,5,3)$, and let $q$ be any prime power. 
Consider the binary code
 $$C:=\{1111100000, \ 1100011100, \ 0011011010, \ 1000110011, \ 0010101101, \
0101000111\}.$$ Observe that $C$ has constant weight $5$ and minimum distance
$6$. Let $v_1,...,v_6$ be the elements of $C$ in the displayed order. It follows from
Theorem \ref{teo1} that:
\begin{enumerate}
 \item $\mbox{MaxDim}_{\underline{\delta}}(\mF_{v_1},\F_q)=15$,
 \item $\mbox{MaxDim}_{\underline{\delta}}(\mF_{v_2},\F_q)=6$,
 \item $\mbox{MaxDim}_{\underline{\delta}}(\mF_{v_3},\F_q)=2$.
\end{enumerate}
Notice moreover that $\mF_{v_4}$ has the following
graphical representation.
$$\begin{array}{ccccc} \mbox{\textbullet} & \mbox{\textbullet} &
\mbox{\textbullet} & \mbox{\textbullet} & \mbox{\textbullet} \\

 & & & \mbox{\textbullet} & \mbox{\textbullet}  \\

& & & \mbox{\textbullet} & \mbox{\textbullet}  \\
 \end{array}$$
By Theorem \ref{rettang} there exist a $2$-dimensional
$\underline{2}$-space $V$ of $2 \times 2$ matrices over $\F_q$ and  a
$2$-dimensional $\underline{1}$-space $W$ of $1 \times 2$ matrices over $\F_q$.
Let $\{ M_1,M_2\}$ and $\{N_1,N_2 \}$ be bases for $V$ and $W$, respectively.
Then
$$\mbox{Span}_{\F_q} \left\{ \begin{bmatrix} 0 & N_i & 0_{1 \times 2} \\ 0 &
0_{2\times 2} & M_i \end{bmatrix} \ | \ i=1,2 \right\} $$
is a $2$-dimensional $\underline{\delta}$-space of matrices with shape
$\mF_{v_4}$. Since $T_{\delta}(\mF_{v_4})=2$, it follows that
$\mbox{MaxDim}_{\underline{\delta}}(\mF_{v_4},\F_q)=2$. Finally, $\mF_{v_5}$
has $T_{\delta}(\mF_{v_5}) =1$ and contains the Ferrers diagram
$$\begin{array}{ccc} \mbox{\textbullet} & \mbox{\textbullet} &
\mbox{\textbullet} \\
 & \mbox{\textbullet} & \mbox{\textbullet} \\

 &  & \mbox{\textbullet}  \\
 \end{array}$$
Hence by Corollary \ref{teo5} and Lemma \ref{tutto} we have
$\mbox{MaxDim}_{\underline{\delta}}(\mF_{v_5},\F_q)=1$.
Using Theorem \ref{boundmulti} we obtain a subspace code $\mC \subseteq
\mG_q(5,10)$ of minimum distance $\delta=3$ with
$$|\mC|=q^{15}+q^6+2q^2+q+1.$$
For $q \ge 3$ this is the subspace code of parameters
$(n,k,\delta)=(10,5,3)$ with largest known cardinality.
\end{example}

Let us briefly recall the definition of \textbf{lexicode}.
The vectors of $\F_2^n$ can be lexicographically ordered as follows. Let $v,w
\in \F_2^n$, $v \neq w$, and let $i:= \min\{ j \ | \ v_j \neq w_j \}$. We say
that $w \prec v$ if $v_i=1$.  Given a binary vector $v \in \F_2^n$ of weight $k$, the
constant weight lexicode originated by $v$ of minimum distance $2\delta$ is
constructed through iterated steps as follows. Start with $C=\{v\}$. List the
elements of $\F_2^n$ in decreasing lexicographic order. At each step add to $C$ the first
vector of the list of weight $k$ and Hamming distance at least $2\delta$
from all the elements of $C$, until there is no such vector left.

According to Theorem \ref{boundmulti}, the cardinality of a subspace code
obtained
through the multilevel construction depends on the choice of the
binary constant weight code. Since lexicodes are known to have large
cardinality
among constant weight binary codes with the same parameters, T. Etzion and
N. Silberstein suggest in \cite{ref1} to use the lexicode originated by the
vector $$\underbrace{1 \cdots 1}_k\underbrace{0 \cdots 0}_{n-k}$$
in the multilevel construction.  
However this choice is not always optimal, as we show in the
following example.

\begin{example} \label{exA}
Let $n:=10$, $k:=5$, $\delta=3$. Consider the  binary constant weight lexicode
$$C':=\{ 1111100000,\  1100011100, \ 1010010011, \ 0101001011, \ 00010101110,
 \ 0001110101\}.$$
Let $w_1,...,w_6$ be the elements of $C'$ in the displayed order. The graphical
representations of the $\mF_{w_i}$'s are as follows.

$$\begin{array}{ccccc} \mbox{\textbullet} & \mbox{\textbullet} &
\mbox{\textbullet} & \mbox{\textbullet} & \mbox{\textbullet}\\

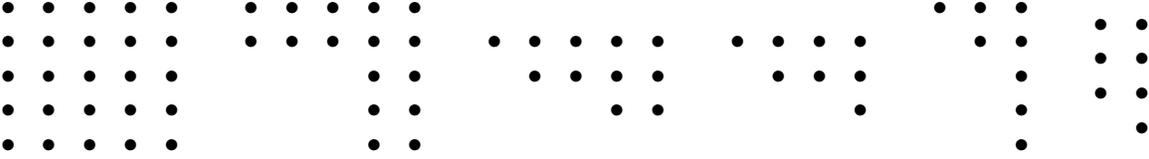
\mbox{\textbullet} & \mbox{\textbullet} &
\mbox{\textbullet} & \mbox{\textbullet} & \mbox{\textbullet}\\

\mbox{\textbullet} & \mbox{\textbullet} &
\mbox{\textbullet} & \mbox{\textbullet} & \mbox{\textbullet}\\

\mbox{\textbullet} & \mbox{\textbullet} &
\mbox{\textbullet} & \mbox{\textbullet} & \mbox{\textbullet}\\

\mbox{\textbullet} & \mbox{\textbullet} &
\mbox{\textbullet} & \mbox{\textbullet} & \mbox{\textbullet}
 \end{array} \ \ \ \ 
 \begin{array}{ccccc} \mbox{\textbullet} & \mbox{\textbullet} &
\mbox{\textbullet} & \mbox{\textbullet} & \mbox{\textbullet}\\

\mbox{\textbullet} & \mbox{\textbullet} &
\mbox{\textbullet} & \mbox{\textbullet} & \mbox{\textbullet}\\

 & & & \mbox{\textbullet} & \mbox{\textbullet}\\

& & & \mbox{\textbullet} & \mbox{\textbullet}\\

& & & \mbox{\textbullet} & \mbox{\textbullet}
 \end{array} \ \ \ \  
 \begin{array}{ccccc} \mbox{\textbullet} & \mbox{\textbullet} &
\mbox{\textbullet} & \mbox{\textbullet} & \mbox{\textbullet}\\

 & \mbox{\textbullet} & \mbox{\textbullet} &
\mbox{\textbullet} & \mbox{\textbullet} \\

 & & & \mbox{\textbullet} & \mbox{\textbullet} 
 \end{array} \ \ \ \ 
 \begin{array}{cccc} \mbox{\textbullet} & \mbox{\textbullet} &
\mbox{\textbullet} & \mbox{\textbullet} \\

 & \mbox{\textbullet} & \mbox{\textbullet} &
\mbox{\textbullet}  \\

 & & &  \mbox{\textbullet} 
 \end{array} \ \ \ \ 
  \begin{array}{ccc}  
\mbox{\textbullet} & \mbox{\textbullet} & \mbox{\textbullet}\\

  &
\mbox{\textbullet} & \mbox{\textbullet} \\

  &  & \mbox{\textbullet} \\
 
  &  & \mbox{\textbullet} \\
 
  &  & \mbox{\textbullet} \\ 
 \end{array} \ \ \ \ 
 \begin{array}{cc}   \mbox{\textbullet} & \mbox{\textbullet}\\
\mbox{\textbullet} & \mbox{\textbullet} \\
\mbox{\textbullet} & \mbox{\textbullet} \\
 & \mbox{\textbullet} 
 \end{array} 
  $$
One can easily check that
$$T_3(\mF_{w_1}) = 15, \ \ \  T_3(\mF_{w_2}) = 6, \ \ \ T_3(\mF_{w_3}) =
2, \ \ \ T_3(\mF_{w_4}) = 1, \ \ \ T_3(\mF_{w_5}) = 1, \ \ \ T_3(\mF_{w_6})
=0.$$
Therefore, by Theorem \ref{boundmulti} and Theorem \ref{tbound}, choosing $C'$
as the pivot
code produces a subspace code of cardinality at most
$$q^{15}+q^6+q^2+2q+1.$$
However, the binary code $C$ considered in Example \ref{exB} produces a
subspace code with the same parameters $n,k,\delta$ and larger cardinality, for
all values of $q$.
\end{example}

Theorem \ref{teo1} allows us to give a lower bound the cardinality of subspace codes
obtained from given pivot vectors through the multilevel
construction.

\begin{theorem}
Fix integers $n,k,\delta$ with $2 \le \delta \le k \le n/2$. Let $D
\subseteq \F_2^n$ be a code of constant weight $k$ and minimum distance at least
$2\delta$. For $v\in D$ let $p_i(v)$ denote the position of the $i$-th one of $v$.
Let $$D':=\{v\in D\mid p_{\delta-1}(v)\leq n-2k+\delta-1\}$$ and 
$$D''=\{v\in D\mid p_i(v)\leq n-k-\delta+2i-1,\: i=1,\ldots,\delta\}.$$ Then
there
exists a subspace code $\mC' \subseteq \mG_q(k,n)$ of injection distance at
least $\delta$ and $$|\mC'| = \sum_{v \in D'} q^{T_\delta(\mF_v)}+\sum_{v \in D''\setminus D'} q+|D\setminus D''|.$$ Moreover
any subspace code $\mC$ obtained from $D$ through the multilevel
construction has $$|\mC| \le  \sum_{v \in D''}q^{T_\delta(\mF_v)}+|D\setminus D''|\leq |\mC'|+ \mathcal{O} \left(
q^{(k-\delta+1)(k-1)} \right)$$
asymptotically in $q$. \end{theorem}

\begin{proof}
For $v \in D$ we denote by $r_i(v)$ the cardinality of the
$i$-th row of $\mF_v$ for $i=1,...,k$. Fix any $v \in D'$. As in Notation
\ref{assoc}, the cardinality of the $i$-th row of $\mF_v$ is
$r_i(v)=n-k-p_i(v)+i$. 

Let $v\not\in D''$, then there exists $i\in\{1,\ldots,\delta\}$ 
such that $r_i(v)\leq n-k-(n-k-\delta+2i)+i=\delta-i$.
Therefore, $$T_{\delta}(\mF_v,i-1)=\sum_{u=i}^k \max\{r_u(v)-(\delta-u+1),0\}=0=T_{\delta}(\mF_v).$$
This proves that any subspace code $\mC$ obtained from $D$ through the multilevel
construction has $$|\mC| \le  \sum_{v \in D''}q^{T_\delta(\mF_v)}+|D\setminus D''|.$$

Let now $v\in D'$. Since $p_{\delta-1}(v) \le n-2k+\delta-1$, we have $r_{\delta-1}(v) \ge k$. 
Combining Theorem \ref{teo1} and the multilevel construction using the vectors of $D'$, 
we construct a code of cardinality $\sum_{v \in D'} q^{T_\delta(\mF_v)}$ and minimum 
distance at least $\delta$. For any $v\in D''$, the associated Ferrers diagram 
$\mF_v\supseteq [\delta,\delta-1,\ldots,1,0,\ldots,0]$ by the definition of $D''$. 
Hence we have at least one matrix of rank $\delta$ and shape $\mF_v$, namely the $k\times (n-k)$ 
matrix containing a top-right justified $\delta\times\delta$ identity matrix and zeroes everywhere else.
Adding the lift of these codewords to the previous code through the multilevel construction, we construct a code $\mC'$ with 
$$|\mC'| = \sum_{v \in D'} q^{T_\delta(\mF_v)}+\sum_{v \in D''\setminus D'} q+|D\setminus D''|$$
as claimed.

It follows from the previous argument that the cardinality of a code $\mC$ obtained from $D$ through the multilevel
construction can be increased only by producing larger linear spaces of matrices of rank at least $\delta$ and support 
contained in $\mF_v$ for $v\in D''\setminus D'$.
Observe that for any such $v$ we have $r_{\delta-1}(v)\leq k-1$, hence $r_i(v)\leq k-1$ for $i=\delta,...,k$. 
Hence $$T_\delta(\mF_v) \le \sum_{i=\delta}^k
r_i(v)\le (k-\delta+1)(k-1).$$ Since $|D\setminus D''|$ and $|D''\setminus D'|$ are constant in $k$, we have
$$|\mC|-|\mC'|\in \mathcal{O} \left(q^{(k-\delta+1)(k-1)}\right).$$
\end{proof}

\begin{example}
 In Table \ref{tabella} we give some cardinalities of subspace
codes which we find combining the results of Section \ref{varicasi}
with the multilevel construction of \cite{ref1} as shown in Example
\ref{exB}. For $q \ge 3$ and the given values of $k$, $n$ and $\delta$,
the codes have the largest known size.

\begin{table}[h!]
 \begin{tabular}{|c|c|c|c|c|}
 \hline
$n$ & $k$ & $\delta$ & size  \\
\hline
\hline
$10$ & $5$ & $3$ & $q^{15}+q^6+2q^2+q+1$  \\
\hline
$11$ & $5$ & $3$ & $q^{18}+q^9+q^6+q^4+4q^3+3q^2$  \\
\hline
$14$ & $4$ & $3$ & $q^{20}+q^{14}+q^{10}+q^9+q^8+2(q^6+q^5+q^4)+q^3+q^2$  \\
\hline
$14$ & $5$ & $4$ & $q^{18}+q^{10}+q^3+1$  \\
\hline
$15$ & $6$ & $5$ & $q^{18}+q^5+1$  \\

\hline

\end{tabular}
\label{tabella}
\caption{Some large subspace codes in $\mG_q(k,n)$ with minimum
injection distance at least $\delta$.}
\end{table}

\end{example}


\begin{thebibliography}{99}



\bibitem{del1} P. Delsarte, \emph{Bilinear forms over a finite field, with
applications to coding theory}. Journal of Combinatorial Theory, Series A, 25
(1978), 3, pp. 226 -- 241.

\bibitem{gui} J.-G. Dumas, R. Gow, G. McGuire, J. Sheekey, \emph{Subspaces of
matrices with special rank properties}. Linear Algebra and its Applications, 
433 (2010), 1, pp. 191 -- 202.

\bibitem{ref1} T. Etzion, N. Silberstein, \emph{Error-correcting codes in
projective spaces via
rank-metric
codes and Ferrers diagrams}. IEEE Transactions on Information Theory, 55
(2009), 7, pp. 2909 -- 2919.

\bibitem{large} T. Etzion, N. Silberstein, \emph{Large Constant Dimension
Codes and Lexicodes}. Advances in Mathematics of Communications, 5(2001),
 2, pp. 177 -- 189.

 \bibitem{e-v} T. Etzion, A. Vardy, \emph{Error-Correcting Codes in Projective
Space}. IEEE Transactions on Information Theory, 57 (2011), 2, pp. 1165 -- 1173.

\bibitem{GR2}  E. Gorla, F. Manganiello, J. Rosenthal, \emph{An Algebraic
Approach for Decoding Spread Codes}. 
Advances in Mathematics of Communications 6 (2012), 4, pp. 443 -- 466.

\bibitem{partial}  E. Gorla, A. Ravagnani, \emph{Partial spreads in random
network coding}. 
Finite Fields and Their Applications 26 (2014), pp. 104--115.

 
\bibitem{origine3}  T. Ho, M. M\'{e}dard, R. K\"{o}tter, D. R. Karger, M.
Effros, J. Shi, B. Leong, \emph{A random linear network coding approach to
multicast}. IEEE Transactions on Information Theory, 52 (2006), pp. 4413 --
4430.

\bibitem{koku} Axel Kohnert, Sascha Kurz, \emph{Construction of Large Constant
Dimension Codes with a Prescribed Minimum Distance}. 
Mathematical Methods in Computer Science, 
Lecture Notes in Computer Science (2008), vol. 5393, pp. 31 -- 42.

\bibitem{KK1} R. K\"{o}tter, F. R. Kschischang, \emph{Coding for Errors and
Erasures in Random Network
Coding}.  IEEE Transactions on Information Theory, 54 (2008), 8, pp. 3579 --
3591. 

\bibitem{origine2} S.-Y.R. Li, R.W. Yeung, N. Cai, \emph{Linear network coding}.
 IEEE Transactions on Information Theory, 49 (2003), 2, pp. 371 --
381.

\bibitem{MS} F. J. MacWilliams, N. J. A. Sloane, \emph{The Theory of
Error-Correcting Codes}. North Holland Mathematical Library.


\bibitem{GR} F. Manganiello, E. Gorla, J. Rosenthal, \emph{Spread Codes and
Spread Decoding
in Network Coding}. IEEE Proceedings (Toronto 2008), pp. 881 --
885.

\bibitem{Mesh} R. Meshulam, \emph{On the maximal rank in a subspace of
matrices}.  Quarterly Journal of Mathematics, 36 (1985), pp. 225 -- 229.

\bibitem{roth}  R. M. Roth, \emph{Maximum-rank array codes and their application to crisscross error correction}.
Transactions on Information Theory, 37 (1991), 2, pp. 328–-336.

\bibitem{alns} N. Silberstein, A.-L. Trautmann, \emph{New lower bounds for
constant dimension codes}. ISIT 2013, pp. 514 -- 518.

\bibitem{alns2} N. Silberstein, A.-L. Trautmann, \emph{Subspace Codes based on
Graph Matchings,
Ferrers Diagrams and Pending Blocks
}. \url{http://arxiv.org/abs/1404.6723}.

\bibitem{aj} A-L. Trautmann, Joachim Rosenthal, \emph{New Improvements on the
Echelon-Ferrers Construction}. Proceedings of the 19th International Symposium
on Mathematical Theory of Networks and Systems - MTNS 2010, pp. 405 -- 408.

\bibitem{tuvi} A. Wachter-Zeh, T. Etzion, \emph{Optimal Ferrers Diagram Rank-Metric Codes}.
\url{http://arxiv.org/abs/1405.1885}.

\end{thebibliography}
\end{document}